\documentclass[11pt]{article}
\usepackage[dvipsnames,usenames]{color}
\usepackage{amsfonts,amssymb,amsthm,mathtools}
\usepackage{paralist}
\usepackage[ruled, noend, noline, linesnumbered]{algorithm2e}
\usepackage{algpseudocode}
\usepackage{bm}
\usepackage{xspace}
\usepackage{xspace,prettyref}
\usepackage{framed}
\usepackage{caption}
\usepackage{subcaption}
\usepackage{color}
\usepackage{mathbbol}
\usepackage[pagebackref,letterpaper=true,colorlinks=true,pdfpagemode=none,urlcolor=blue,linkcolor=blue,citecolor=BrickRed,pdfstartview=FitH]{hyperref}
\usepackage{fullpage}

\newtheorem{theorem}{Theorem}

\newtheorem{lemma}[theorem]{Lemma}

\newtheorem{definition}{Definition}
\theoremstyle{definition}

\newcommand{\EX}{\hbox{\bf E}}

\def\eps{\varepsilon}

\def\floor#1{\lfloor {#1} \rfloor}

        {\hspace*{\fill}$\Box$\par}

\newcommand{\NN}{\mathbb{N}}

\newcommand{\tF}{\widetilde{F}}
\newcommand{\tG}{\widetilde{G}}

\newcommand{\Sec}[1]{\hyperref[sec:#1]{\S\ref*{sec:#1}}} %section
\newcommand{\Eqn}[1]{\hyperref[eqn:#1]{(\ref*{eqn:#1})}} %equation
\newcommand{\Fig}[1]{\hyperref[fig:#1]{Fig.\,\ref*{fig:#1}}} %figure
\newcommand{\Tab}[1]{\hyperref[tab:#1]{Tab.\,\ref*{tab:#1}}} %table
\newcommand{\Thm}[1]{\hyperref[thm:#1]{Theorem\,\ref*{thm:#1}}} %theorem
\newcommand{\Lem}[1]{\hyperref[lem:#1]{Lemma\,\ref*{lem:#1}}} %lemma
\newcommand{\Prop}[1]{\hyperref[prop:#1]{Prop.~\ref*{prop:#1}}} %property
\newcommand{\Cor}[1]{\hyperref[cor:#1]{Corollary~\ref*{cor:#1}}} %corollary
\newcommand{\Def}[1]{\hyperref[def:#1]{Definition~\ref*{def:#1}}} %definition
\newcommand{\Alg}[1]{\hyperref[alg:#1]{Alg.~\ref*{alg:#1}}} %algorithm
\newcommand{\Ex}[1]{\hyperref[ex:#1]{Ex.~\ref*{ex:#1}}} %example
\newcommand{\Clm}[1]{\hyperref[clm:#1]{Claim~\ref*{clm:#1}}} %example

\newcommand{\hash}{\hbox{hash}}

\newcommand{\C}{\hbox{C}}
\newcommand{\tC}{\widetilde{\hbox{C}}}
\newcommand{\ct}{\hbox{ct}}
\newcommand{\tg}{\hbox{TG}}
\newcommand{\cdf}{\hbox{C}}

\newcommand{\rd}{{\hbox{red}}}
\newcommand{\loss}{{\hbox{loss}}}
\newcommand{\hN}{\widehat{N}}
\newcommand{\mult}{k}

\newcommand{\degdist}{{\tt headtail}}
\newcommand{\update}{{\tt update}}
\newcommand{\est}{{\tt estimate}}

\author{Olivia Simpson\thanks{Work was done while the author
was an intern at Sandia National Laboratories, Livermore.}
\\ {\tt osimpson@ucsd.edu}\\
University of California, San Diego
\and C. Seshadhri
 \\ {\tt scomandu@ucsc.edu}\\
University of California, Santa Cruz
\and Andrew McGregor
\\ {\tt mcgregor@cs.umass.edu}\\
University of Massachusetts, Amherst
}

\title{Catching the head, tail, and everything in between: a streaming algorithm for the degree distribution}
\date{}

\begin{document}

\maketitle

\begin{abstract}
%\boldmath
The degree distribution is one of the most fundamental graph properties of interest
for real-world graphs. It has been widely observed in numerous domains that graphs
typically have a \emph{tailed} or \emph{scale-free} degree distribution. While the average
degree is usually quite small, the variance is quite high and there are vertices
with degrees at all scales. We focus on the problem of approximating the degree distribution
of a large streaming graph, with small storage. We design an algorithm \degdist{}, whose main novelty
is a new estimator of infrequent degrees using truncated geometric random
variables. We give a mathematical analysis of \degdist{}
and show that it has excellent behavior in practice. We can process streams with millions
of edges with storage less than $1\%$ and get extremely accurate approximations
for \emph{all} scales in the degree distribution.

We also introduce a new notion of \emph{Relative Hausdorff} distance between tailed
histograms. Existing notions of distances between distributions are not suitable, since they ignore
infrequent degrees in the tail. The Relative Hausdorff distance measures deviations at all scales,
and is a more suitable distance for comparing degree distributions. By tracking this new measure,
we are able to give strong empirical evidence of the convergence of \degdist.
\end{abstract}

\section{Introduction}

Graphs are a natural abstraction for any data set with entities and relationship between them. Popular examples include online social networks such as Facebook and Twitter;  transportation networks; biological networks such as protein-protein interaction and metabolic networks; and communication networks such as the internet and telephone and email networks.
Many of these graphs are most naturally represented by
a \emph{stream of edges}. Especially for social and communication networks, each edge has an associated timestamp, and the graph
is basically an aggregate of all these edges over some time window.
Such streams are typically quite massive; social networks like Facebook and Twitter can generate billions of communication links in a day~\cite{KwLe10,FB}.
A publicly available HTTP request dataset has billions of requests~\cite{Meiss08WSDM}. The scale of these data sizes has led
to interest in \emph{small-space streaming algorithms}. Such algorithms accurately compute specific properties
of the total graph, using a memory footprint that is orders of magnitude smaller in size.

Arguably, one of the most important properties of real-world networks is the \emph{degree distribution}. Seminal papers
in massive graph analysis studied precisely this quantity~\cite{BarabasiAlbert99,FFF99,BrKu+00}.
The study of degree distributions
is probably the birthplace of real-world network analysis. It has been found to be relevant for graph modeling, network resilience, and algorithmics~\cite{CoEr+00,NeStWa01,PeFlLa+02,Ne03,Mi03,ChFa06,SeKoPi11}.
One of the key discoveries of network analysis is the presence
of \emph{scale-free} or heavy-tailed degree distributions. The average degree of a node is usually small,
but there are nodes with degrees at all scales. The very notion of a \emph{scale-free network} has entered the common parlance
because of its relevance to network analysis~\cite{Wiki-scale}.

\subsection{Problem statement} \label{sec:problem}

The input is a stream of edges $e_1, e_2, \ldots, e_m$ without any repetitions.
The graph created by these edges is denoted $G = (V,E)$. For convenience,
we set $V = [n]$, though the labels may be from some arbitrary discrete universe.
We do not assume that the algorithm knows $n$ and $m$, the number of vertices
and edges respectively.
Each edge is represented by a pair $(u,v)$ of vertex labels.

For vertex $v \in V$, $d_v$ denotes its degree (the number of neighbors of $v$).
We set $n(d)$ to be the number of vertices of degree $d$, and $N(d)$ to be the number
of vertices of degree at least $d$. In math, $N(d) = \sum_{r \geq d} n_r$.
It is convenient for us to work with unnormalized raw counts, so we deal
with histograms rather than distributions.
We denote the sequence $\{n(d)\}$ by the \emph{degree histogram} (dh) and
$\{N(d)\}$ is the \emph{complementary cumulative degree histogram}\footnote{This is often
called the cumulative degree distribution, but that is counter to the standard
definition for probability distributions.} (ccdh).
When $\{n(d)\}$ is normalized by $n$, it is called the degree distribution.
We focus on the ccdh, instead of the dh. Typically, the dh is quite noisy in real data, and
the ccdh has the added benefit of being monotonically decreasing. (Focus on the ccdh is standard
for fitting procedures~\cite{ClShNe09}.)

We study the problem of approximating the ccdh of $G$ using a \emph{small-space one-pass streaming algorithm}.
Such an algorithm has some limited memory, denoted $M$. It sees the edges in stream order,
and on seeing edge $e_t$, updates the memory $M$. The algorithm cannot access older edges,
and $M$ is typically order of magnitudes smaller than the size of the stream. At the end of the stream,
the algorithm reports a sequence $\{\hN(d)\}$, an approximation to the ccdh of $G$.

We make no assumption on the ordering of edges.
We do not consider edge deletions or edge repetitions. (This is the standard
model used in most work on practical streaming algorithms.)

% This has numerous applications including community detection  in a social network and finding influential nodes suitable for targeted advertising; identifying promising therapeutic targets in biological networks; and designing better algorithms for the infrastructure supporting man-made networks such as routing protocols and web search.
%
% Existing work focuses mainly on estimating aggregate statistics such as the number of triangles \cite{jha2013space} and the related transitivity coefficient or partitioning the graph based on the sizes of cuts and more general spectral properties \cite{AhnCGMW15,PavanTTW13}.
% We present an algorithm for estimating the degree distribution of a graph. The degree distribution captures important about the graph including the fragility of the network to noise and \ldots. Barab\'asi-Albert model. Work in the past has considered fitting parameters \cite{} to specific models for the degree distribution, the most common being power-law distributions.
%
%
% Let $G$ be a simple undirected graph given by a stream of edges $\sigma = (e_1,
% e_2, \ldots, e_m)$ in some arbitrary order.  We are allowed one pass over the
% stream for our computation, and we store information about $s$ sampled vertices
% from the stream.

\begin{figure*}[t]
    \centering
    \begin{subfigure}[b]{0.3\textwidth}
        \centering \includegraphics[width=\textwidth]{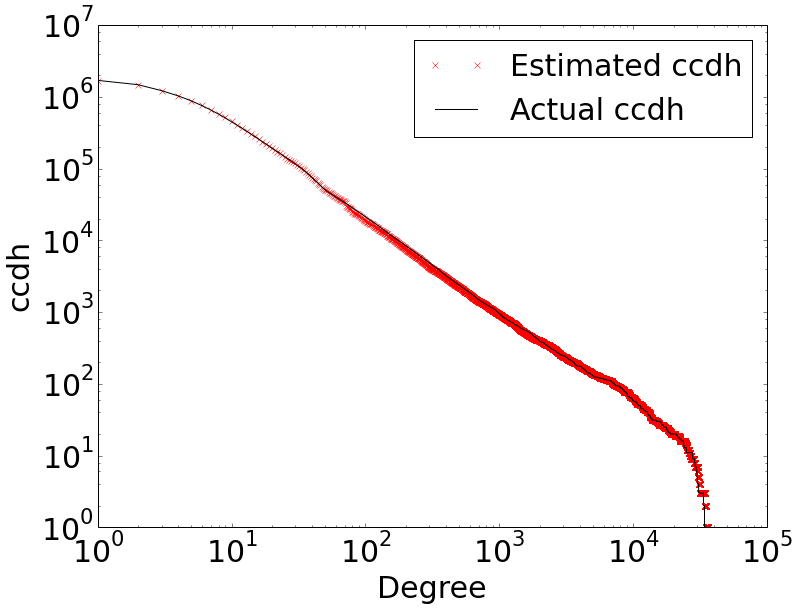}
        \caption{as-Skitter: $n = 1.7M$, $m = 11M$, storage $=31K$}
        \label{fig:skitter-ccdh}
    \end{subfigure}
    \hfill
    \begin{subfigure}[b]{0.3\textwidth}
        \centering \includegraphics[width=\textwidth]{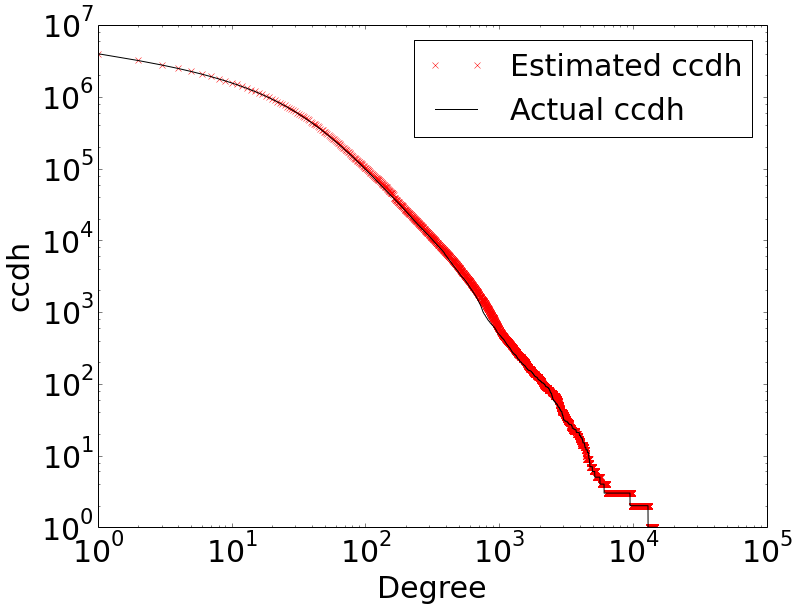}
        \caption{com-LiveJournal: $n = 4M$, $m = 34M$, storage $=200K$}
        \label{fig:livejournal_tailspace_04}
        \end{subfigure}
    \hfill
    \begin{subfigure}[b]{0.3\textwidth}
        \centering \includegraphics[width=\textwidth]{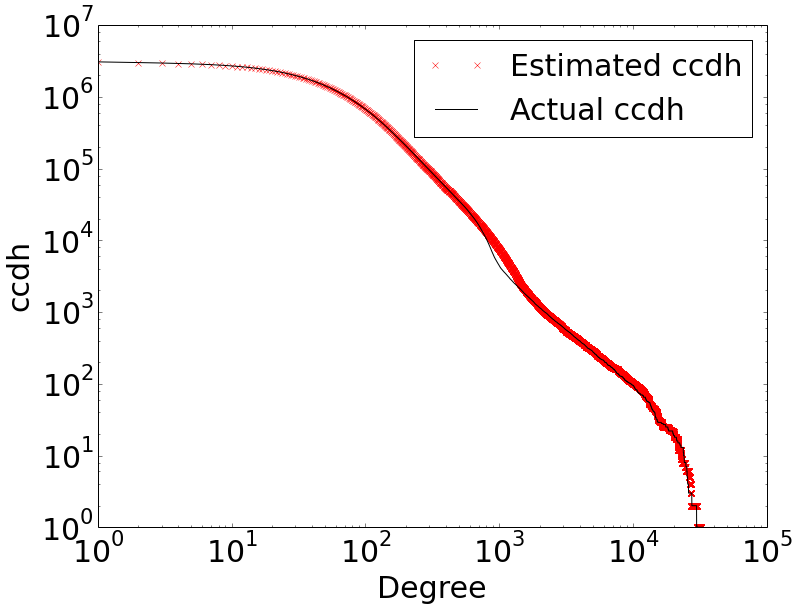}
        \caption{com-Orkut: $n = 3M$, $m = 117M$, storage $=150K$}
        \label{fig:orkut_tailspace_04}
    \end{subfigure}
    \caption{The output ccdh of \degdist{} on three different input graphs from
the SNAP~\cite{snap} collection. In each case, the storage is less $1\%$ of the stream
    (and less than $5\%$ of the number of vertices). Observe the near identical match
    with the true ccdh.}
    \label{fig:ccdh}
\end{figure*}

\subsection{Challenges}

\textbf{How does a small-space algorithm estimate the degree distribution at all scales?}
 The degree distribution involves degrees at ``all" scales:
many low degree vertices, some intermediate degree vertices, and few very high degree vertices.
Look at~\Fig{skitter-ccdh} for the ccdh of a router topology network. The average degree is $20$,
but there are vertices with degrees up to $50,000$.
The count of low degree vertices is easy to estimate, since a simple random sample of vertices gives a good estimate.
Intermediate and high degrees pose a problem. There are few such vertices but it is critical to sample
their count accurately. There is a huge literature on estimating distribution properties
of a stream of items: frequent items, distribution moments, distinct items, etc.~\cite{IndykW05,KaneNW10,cormode2008frequent}.
(We discuss in depth later.) But these only give specific properties of the distribution. None of these methods can get frequency estimates at \emph{all scales}, ranging continuously from (frequent)
low degrees to (infrequent) high degrees.

\textbf{How to quantitatively compare (cumulative) degree distributions?} How do we actually assert that our algorithm
is any good?
One can use standard statistical distance measures like Kolmogorov-Smirnov.
Yet these measures typically ignore
the tail since it contains a negligible fraction of vertices. Consider the following examples. We take a clique
of $n$ vertices and a clique of $n-1$ vertices. It is natural to say that their degree distributions are quite close,
but no popular existing measure would assert that. On the other end, consider a star with $n$ edges, and a matching with $n$
edges. The degree distribution only differs at one ``point", the vertex of degree $n$. Yet we would consider
the degree distributions to be fundamentally different. Most statistical measures would say they are similar, since
they differ at only a single outlier.

An intuitive notion of similarity is closeness in log-log plots, but how do we quantify such a concept? One might
try to approximate degree distributions by closed-form, but fitting procedures are notoriously tricky
for tailed distributions and subject to much error~\cite{ClShNe09}.

\subsection{Main results}
%total number of samples: 31760 = 0.0187218339852 * n

\textbf{The algorithm \degdist:} Our main contribution is a new small-space algorithm \degdist{} that estimates
the ccdh of an input graph stream. The novelty is a new estimator for infrequent
degree counts, which is combined with standard sampling to give
ccdh estimates at all scales. We represent the sampling of \degdist{} through certain \emph{truncated geometric}
random variables. An analysis of their behavior provides the right ``correction" factors
to infer the ccdh from our sampling.
We provide a detailed mathematical analysis of \degdist{} explaining why it accurately
estimates the ccdh. Our analysis falls short of a complete proof, and we rely
on some heuristic arguments for the full argument.

\textbf{Relative Hausdorff distance:} We introduce a new notion of distance between
ccdhs (technically, between any two histograms) called the \emph{Relative Hausdorff (RH) distance}.
This distance avoids the pitfalls of standard measures, and is able to capture the closeness
at \emph{all scales}. Intuitively, a small RH-distance implies that \emph{every} point in one ccdh
is ``close" (up to relative error) to some point in the other ccdh. Put another way, both ccdhs
agree at all scales, and agree on outliers. While this condition is quite stringent, RH distance
is flexible enough to allow for minor errors. It gives a concrete way of quantifying the quality
of \degdist, and empirically establishing convergence of our estimate.

\textbf{Empirical behavior of \degdist:} We run \degdist{} on a wide variety of public graph datasets.
It gives excellent estimates of the ccdh in all our tests, for storage less than 1\% of the stream.
We show example outputs in \Fig{ccdh}, for three different input graphs. In each case, observe
the near perfect match with the true ccdh, at \emph{all} degrees.
We compute the RH distance for numerous runs and demonstrate convergence of \degdist{}'s output
with increasing storage. In all our runs, storage around 1\% of the stream is sufficient for
excellent match in ccdhs (and also for low RH-distance).

\subsection{Related Work}
\label{sec:relatedwork}
Note that we can frame our problem in terms of general histogram estimation.
If one views the input as a stream of vertex labels, then the dh (and ccdh)
is the histogram of label frequencies. There is much work on understanding frequencies in a discrete
stream, but as we detail below, none of this work solves the problem of estimating the ccdh.

Finding frequent items, aka ``heavy hitters," is a classic problem in the data stream model.
Cormode and Hadjieleftheriou \cite{cormode2008frequent} compare three of the most important algorithms: the
\emph{frequent} algorithm~\cite{demaine2002frequency, karp2003simple,BerindeICS10}, the
\emph{lossy counting} algorithm~\cite{manku2002approximate}, and the \emph{space
saving} algorithm~\cite{metwally2005efficient}.\footnote{Other popular algorithms such as CountSketch~\cite{charikar2002finding} and
CountMin~\cite{cormode2005improved} enable frequent items to be identified when the frequency of an item may be incremented and decremented.}
For large degrees, these approaches will give accurate results, but the error term dwarfs the degree at smaller scales. We demonstrate this empirically in Section \ref{sec:expresults}.
Much work has been done in approximating frequency moments~\cite{AlonMS99,IndykW05,KaneNW10,cormode2008frequent}, but they do not
give an estimate for multiple scales. Nor has this work been implemented in practice for large data sets.

Rather than just finding frequent items, Korn et al.~\cite{korn2006modeling} attempt to estimate the entire distribution of elements in the stream. However, in contrast to our work, their approach assumes that the distribution comes from a parameterized family of distributions, e.g., the distribution is Zipfian, and then focuses on estimating the relevant parameters. This approach is only applicable for graphs where the degree distribution is already relatively well understood.
Despite much study and claims, there are \emph{no} conclusive closed-form formulae for real-world degree distributions.
The classic power law fitting work of Clauset et al.~\cite{ClShNe09} argues why most previous methods are not statistically robust,
and how one needs strong independence assumptions to get rigorous results. Therefore, \degdist{} makes no closed form assumption
on the input stream.

Over the last ten years, there has been a growing body of work focused on processing graphs in the data stream model. See~\cite{mcgregor2014graph} for a summary of recent work on graph streaming and sketching. This work has included problems such as the number of triangles and related quantities such as the transitivity coefficient \cite{jha2013space,PavanTTW13,ahmed2014graph}, estimating the connectivity properties of a graph \cite{GuhaMT15}, and solving combinatorial problems such as computing large matchings \cite{KapralovKS14,McGregor05}.
Cormode and Muthukrishnan considered estimating properties of the degree
distribution in multigraphs but not the distribution itself\cite{CormodeM05}.

Closest to this work is the series of graph sampling papers by Ahmed et al.~\cite{ayman2013ads,ahmed12socialnets,ahmed2014network,ahmed2014graph}.
Their work focuses on estimating many properties (as opposed to a single property) with a fixed sampling method, and they study various sampling schemes.
The results on estimating ccdhs typically use 20-30\% of the stream, with weaker empirical results~\cite{ayman2013ads}.
The recent Graph Sample and Hold framework gives extremely strong results for triangle counting~\cite{ahmed2014graph}, but is
not applied for the ccdh.
This technique is closely related to an approach for estimating frequency moments
\cite{AlonMS99,braverman2013approximating}. Our sampling approach is also similar, and our main
contribution is in the actual estimation procedure.

\section{The algorithm}

The algorithm \degdist{} has two parts: \update{} and \est.
The procedure \update{} is called for every edge in the stream, and simply
updates the data structures. The procedure \est{} is called at the end
of the stream to get an estimate of $\{N(d)\}$. In what follows,
the subscript $h$ refers to ``head" and $t$ is ``tail".

The algorithm \degdist{} requires two parameters, $p_h$ and $p_t$, which
are probabilities. These decide the storage requirements of the algorithm,
as explained later. For convenience, we will assume these are global variables,
and will not pass them around to each function.

We will assume the existence of a hash function $hash$ that maps strings
uniformly to $[0,1]$.

{\bf Data Structures:} There are two sets of vertices $S_h$ and $S_t$,
and corresponding maps $\ct_h:S_h \mapsto \NN$ and $\ct_t:S_t \mapsto \NN$.
Again, we assume these are global variables.

{\bf The procedure \update:} This updates the data structures for each edge in the stream.
Consider edge $(u,v)$ in the stream.
If $v \in S_h$, the $\ct_h(v)$ is incremented (analogously for $S_t$). Now for the critical
difference between $S_h$ and $S_t$. If $v \notin S_h$ and if $\hash(v) \leq p_h$,
then $v$ is added to $S_h$. If $v \notin S_t$: we insert $v$ to $S_t$ with probability $p_t$.
(The entire operation above is also done for $u$.)
Note the difference: for $S_h$, we essentially flip a random coin for the vertex. For $S_t$, we flip a coin for the edge.
Intuitively, $S_h$ is maintaining a uniform random set of vertices. On the other hand,
$S_t$ maintains sample of vertices biased towards higher degree.

{\bf The procedure \est:} This procedure uses $S_h, S_t, \ct_h, \ct_t$ to output
an estimate $\{\hN(d)\}$ for the ccdh of $G$. We set $\C_h(r)$ to be the number
of vertices in $S_h$ with $\ct_h(\cdot)$ value of $r$ (similarly for $\C_t(r)$).
One can think of this as the ``observed" degree distribution. The scaling of
$\C_h(r)$ is straightforward: we simply consider $\C_h(r)/p_h$ to be an estimate
of $n(r)$.  By summing these appropriately, we get an estimate (the head
estimate) of $N(r)$.

For $\C_t$, we first do an additive ``correction". So we set $\tC_t(r) = \C_t(r - \ell(r))$,
where $\ell(r)$ is a correction factor. The explanation of this factor is
provided in Section~\ref{sec:math}.
Then, we do a biased scaling and consider $\tC_t(r)/(1-(1-p_t)^r)$ as an estimate of $n(r)$.
Again, by taking partial sums, we have an estimate (the tail estimate) of $N(r)$.

Observe that we have two different estimates of $N(r)$. We prove in our mathematical analysis
that the former is accurate for the head of the distribution, while the latter is appropriate
for the tail. This distinction is made by $d_{thr}$, which is chosen to ensure
that the first estimate has low variance. Hence, for all degrees less than $d_{thr}$, we
use the head estimate, and for the remaining, we use the tail estimate.

\medskip
We now give a formal description of the algorithm.

%%
%% The basic structure of the algorithm is quite simple. It maintains two sets
%% of vertices $S_h$ and $S_t$, which are updates as the stream proceeds.
%% For each vertex in $v \in S = S_h \cup S_t$, a count $\ct(v)$ is maintained.
%% This process is fairly straightforward, and is shown in the process
%% \update. The tricky part is get estimates
%% for the cumulative degree distribution $F(d)$ from these counts. This
%% is given in \est.
%%
%% Abusing notation, we will use $hash$ to map vertices and edges.
%%

\begin{algorithm}
\caption{\degdist$(p_h,p_t)$} \label{alg:degdist}
\DontPrintSemicolon
Initialize empty sets $S_h$ and $S_t$ and empty mappings $\ct_h$ and $\ct_t$.\;
For each edge $e_i = (u,v)$ in the stream,\;
\ \ \ \ Call \update$(u,v)$.\;
Call \est{} to get output estimate for $\{N(d)\}$.\;
\end{algorithm}

\begin{algorithm}
\caption{\update$(u,v)$} \label{alg:update}
\DontPrintSemicolon
If $u \in S_h$, increment $\ct_h(u)$.\;
If $u \notin S_h$: if $hash[u] < p_h$, insert $u$ in $S_h$ and set $\ct_h(u) = 1$.\;
If $u \in S_t$, increment $\ct_t(u)$.\;
If $u \notin S_t$: with probability $p_t$, insert $u$ in $S_t$ and set $\ct_t(u) = 1$.\;
(Repeat above steps for $v$.)
\end{algorithm}

For fixed $p_t \in (0,1)$, we define $\ell(r)$ to be:
$$\Big\lceil\frac{1 - p_t - (1-p_t)^{r+1} - rp_t(1-p_t)^r}{p_t(1-(1-p_t)^r)}\Big\rceil$$

\begin{algorithm}
\caption{\est} \label{alg:est}
\DontPrintSemicolon
Let $\C_h(r)$ be the number of vertices in $S_h$ with count exactly $r$. (Similarly, define $\C_t(r)$).\;
For all counts $r$, set $\tC_h(r) = \C_h(r-\ell(r))$.\;
For all counts $r$:\;
\ \ \ \ Set $g_h(r) = \C_h(r)/p_h$.\;
\ \ \ \ Set $g_t(r) = \tC_t(r)/[1-(1-p_t)^r]$.\;
Set $d_{thr}$ to be largest $d$ such that $\sum_{r \geq d} g_h(r) \geq 50/p_h$.\;
For all degrees $d$:\;
\ \ \ \ If $d \leq d_{thr}$, set $\hN(d) = \sum_{r \geq d} g_h(r)$.\;
\ \ \ \ If $d > d_{thr}$, set $\hN(d) = \sum_{r \geq d} g_t(r)$.\;
\end{algorithm}

\section{Mathematical Analysis} \label{sec:math}

We abstract out the behavior of the algorithm in a series of claims.
We stress that all our theorems are independent of graph stream order, and hence
\est{} works for all orderings.

\begin{definition} \label{def:geo} For any positive integer $s$ and $p \in (0,1)$, the \emph{truncated
geometric distribution} $\tg_{p,s}$ has the pdf: $\forall 0 \leq k \leq s-1$, $\Pr[X = k] = p(1-p)^k/[1-(1-p)^s]$.
\end{definition}

Observe that as $s \rightarrow \infty$, this is a standard geometric random variable.

\begin{lemma} \label{lem:samp-head} For every $v \in [n]$, $v$ is inserted in $S_h$ independently with probability $p_h$. Conditioned on $v \in S_h$, $\ct(v) = d_v$.
\end{lemma}

\begin{proof} We assume that $\hash$ is a uniform random function, so $\hash(v)$ is uniformly distributed in $(0,1)$.
The probability that $\hash(v) \leq p_h$ is exactly $p_h$. Observe that if $\hash(v) \leq p_h$, then $v$
is inserted in $S_h$ at the very first occurrence of $v$ in the stream. Hence $\ct(v) = d(v)$, whenever $v \in S_h$.
\end{proof}

\begin{lemma} \label{lem:samp-tail} For every $v \in [n]$, $v$ is inserted in $S_t$ independently with probability $1-(1-p_t)^{d_v}$. Conditioned on $v \in S_t$, $\ct(v) = d_v - X$, where $X \sim \tg_{p_t,d_v}$.
\end{lemma}

\begin{proof} There are $d_v$ occurrences of $v$ in the stream. The probability
of $v$ being added in the $b$th occurrence is $p_t(1-p_t)^{b-1}$. When this happens,
$\ct(v) = d_v - (b-1)$. The probability that $v$ is never added is $\sum_{b=1}^{d_v} p_t(1-p_t)^{b-1} = (1-p_t)^{d_v}$.
Conditioned on $v$ being added to $S_t$, the probability of $v$ being added
in the $b$th occurrence is exactly $p_t(1-p_t)^{b-1}/[1-(1-p_t)^{d_v}]$.
So $b-1$ is distributed as $\tg_{p_t,d_v}$.
\end{proof}

\begin{lemma} \label{lem:geo} The expected value of $X \sim \tg_{p,d}$ is $\frac{1 - p - (1-p)^{d+1} - dp(1-p)^d}{p(1-(1-p)^d)}$.
\end{lemma}

\begin{proof} Using the bound for the sum of an arithmetico-geometric series:
\begin{align*}
\frac{p}{1-(1-p)^d} \sum_{k=0}^d k(1-p)^k
&= \frac{p}{1-(1-p)^d} \left( \frac{(1-d)(1-p)^d}{p} + \frac{(1-p)-(1-p)^d}{p^2} \right)\\
&= \frac{1 - p - (1-p)^{d+1} - dp(1-p)^d}{p(1-(1-p)^d)}.
\end{align*}
\end{proof}

This expression is exactly (up to rounding) $\ell(d)$.
Conditioned on $v \in S_t$, $\EX[\ct(v)]$ is $d_v$ minus a ``loss" term, which is precisely
the expression in \Lem{geo}. That should hopefully explain the use of $\ell(d)$ in our algorithm.
We make the (admittedly wrong) assumption that \emph{every} vertex of degree $d$ in $S_t$
``loses" exactly the expected loss. In other words, we assume that $\ct(v)$ is $\EX[\ct(v)]$.
To infer the number of degree $d$ vertices in $S_t$, we add back the expected loss
to each vertex in $v$. That is why we set $\tC_t(r) = \C_t(r-\ell(r))$.

\medskip

It is fairly easy to bound the space and running time of \degdist.

\begin{theorem} \label{thm:time} The expected space used by \degdist{} is $O(p_h n + p_t m)$.
The expected running time of \update{} is $O(1)$, and the expected running time of \est{} is $O(p_h n + p_t m)$.
\end{theorem}

\begin{proof} We will store all sets as hash tables, to ensure $O(1)$ updates.
By \Lem{samp-head}, each vertex is added to $S_h$ with probability $p_h$. Hence, the expected size
of $S_h$ is $O(p_h n)$. For each edge in the stream, we potentially add a vertex
to $S_t$ with probability $p_t$. Hence, the expected size of $S_t$ is $O(p_t m)$. (This is a gross
upper bound, and a refined bound based on \Lem{samp-tail} would be $\sum_d n(d) [1-(1-p)^d]$.)

The processing of \update{} only requires addition in set and count increments,
and requires $O(1)$ time. The procedure \est{} runs in time linear in the sets $S_h$ and $S_t$.
\end{proof}

\subsection{The estimators} \label{sec:est}

For the analysis of our estimators, we need to introduce various error parameters. Natually,
the actual implementation \est{} simply sets these to be fixed constants, so we make
slight modifications and assumptions for convenience of analysis.

Let $\eps = (0,1)$ be an error parameter, and let $c$ be a sufficiently large constant.
\begin{asparaitem}
    \item We set $d_{thr}$ to be the largest $d$ such that $\sum_{r \geq d} g_h(r) \geq (c(\log n)/\eps^2)/p_h$.
    (In the implementation, we hardcoded $c/\eps^2$ to be $50$.)
    \item We assume that $p_t$ is chosen so that $d_{thr} \geq \log(1/\eps)/p_t$.
\end{asparaitem}

\medskip

We begin with the analysis of the head estimator, which is a straightforward Chernoff bound application.

\begin{lemma} \label{lem:head} For all $d \leq d_{thr}$, $\EX[\hN(d)] = N(d)$.
With probability $> 1 - 1/n$, for all $d \leq d_{thr}$, $|\hN(d) - N(d)| \leq \eps N(d)$.
\end{lemma}

\begin{proof} Fix some $d \leq d_{thr}$. Note that the head estimator
is used for $\hN(d)$. Also, $\sum_{r \geq d} g_h(r)$ is precisely the number of vertices of degree at least $d$ in $S_h$.
For convenience, denote this by $X_d$, and observe that it is monotonically decreasing in $d$.
By \Lem{samp-head}, each vertex is added independently to $S_h$ with probability $p_h$.
Thus, $\EX[X_d] = p_h \cdot N(d)$.
Note that $\hN(d)$ is precisely $X_d/p_h$, so $\EX[\hN(d)] = N(d)$.

Since $d_{thr}$ is itself a random variable, we need a little care to prove the lemma.
Observe that $X_d$ is well-defined for all $d$, and is the sum of Bernoulli random
variables. By a multiplicative Chernoff bound (refer to Theorem 1.1 in~\cite{DuPa}), $\Pr[|X_d - \EX[X_d]| \leq \eps \EX[X_d]] \leq 2\exp(-\eps^2 \EX[X_d]/3)$.
Furthermore, by an alternate bound, if $B \geq e\EX[X]$, then $\Pr[X \geq B] < 2^{-B}$.

When $\EX[X_d] = p_h \cdot N(d) \geq (c(\log n)/3\eps^2)$, apply the first bound.
When $\EX[X_d] < (c(\log n)/3\eps^2)$, apply the second bound with $B = c(\log n)/\eps^2$.
Finally, we apply the union bound over all errors, which a calculation shows to be $< 1/n$.
Hence, for any $d$ where $\EX[X_d] < c(\log n)/3\eps^2$, $X_d < c(\log n)/\eps^2$.
So, $d_{thr}$ must be smaller than any such degree. Thus, for all $d \leq d_{thr}$,
$\EX[X_d] \geq c(\log n)/3\eps^2$, and the first Chernoff bound gives the desired
concentration.
\end{proof}

The more challenging part is to analyze the tail estimator. We fall short of giving a complete
proof that it works. Nonetheless, we provide
some mathematical evidence of its correctness. We provide a high level explanation
of the math that follows. We warn the reader that we shall switch between estimates for $N(d)$
and $n(d)$.

The weakness of the head estimator is made clear in the proof of the previous lemma.
The Chernoff bounds says that the error probability of estimating of $N(d)$ is roughly $\exp(-p_h\cdot N(d))$.
This goes to $1$ as $N(d)$ becomes smaller than $1/p_h$. That is precisely what happens
in the tail of the degree distribution, which contains fewer vertices of higher degree.
In general, mild fluctuations in estimates for low degree vertices is ok (there are many of them),
but even a little wagging in the tail estimates creates significant error.

But high degree vertices are more likely to be in $S_t$ by \Lem{samp-tail}.
Let $S_t(d)$ denote the subset of degree $d$ vertices in $S$.
We show in \Lem{highdeg} how to get an estimate of $n(d)$ from $|S_t(d)|$,
where the error probabilities are roughly $\exp(-p_t\cdot d \cdot n(d))$.
\emph{Note the extra $d$ factor}. As long as $d\cdot n(d) \geq 1/p_t$, we can
hope for concentration. In other words, even though high degree vertices are infrequent,
it is provably possible to get accurate estimates for these counts.

Unfortunately, it is not clear how to estimate $|S_t(d)|$, since $\ct_t(v)$
is quite different from $d_v$. As mentioned earlier, we
make the (admittedly erroneous) assumption that $\ct_t(v) = d_v - \EX_{X \sim \tg_{p_t,d_v}}[X]$,
based on \Lem{samp-tail} and \Lem{geo}. This is used to predict the actual degree
of $v \in S_t$, based on $\ct_t(v)$. While this assumption is wrong because the truncated
geometric distribution has large variance, in practice, it works quite well.

In \est, the proxy for $|S_t(d)|$ is given by $\tC_t(d)$. We show
that the ``ccdh" (or partial sums) of $\tC_t(d)$ approximates those of $|S_t(d)|$.
In other words, we can a get a rough approximation for the number of vertices
of degree at least $d$ in $S_t$. This is what is proven in \Thm{tc} and the subsequent calculations.

We now proceed with the formal proofs. The following lemma provides an appropriate concentration bound
for estimating $n(d)$ from $|S_t(d)|$.

\begin{lemma} \label{lem:highdeg} For all $d$, $\EX[|S_t(d)|] = (1-(1-p_t)^d)n(d)$. For all
$d \geq d_{thr}$ and sufficiently small $p_t$: with probability at least $1-2\exp(-\eps p_t\cdot d\cdot n(d)/16)$,
$\Big||S_t(d)| - \EX[|S_t(d)|]\Big| \leq \eps \EX[|S_t(d)|]$.
\end{lemma}

\begin{proof} Every degree $d$ vertex is added to $S$ with probability $1-(1-p_t)^d$ (for convenience,
denote this by $\alpha$).
Linearity of expectation proves that $\EX[|S_t(d)|] = \alpha n(d)$.
Note that $|S_t(d)|$ is the sum of Bernoulli random variables, each with expectation $\alpha$.
By the original Chernoff-Hoeffding bound~\cite{Ch52}, $\Pr[|S_t(d)| \leq (1-\eps)\alpha n(d)] \leq \exp(-D(\alpha(1-\eps) \| \alpha) n(d))$,
where $D(\cdot,\cdot)$ denotes the KL-divergence. With some manipulations,
\begin{align*}
D(\alpha(1-\eps) \| \alpha) 
& = \alpha(1-\eps) \ln\frac{\alpha(1-\eps)}{\alpha} + (1-\alpha(1-\eps)) \ln\frac{1-\alpha(1-\eps)}{1-\alpha} \\
& \geq \alpha\ln(1-\eps) + \alpha \eps \ln(1 + \alpha\eps/(1-\alpha))
\end{align*}
Now we use $d \geq d_{thr} \geq \log(1/\eps)/p_t$. A calculation yields
$[1-(1-p_t)^d]\eps/(1-p_t)^d \geq 1/2$ for sufficiently small $p_t$.
Hence, the expression above is bounded below by:
\begin{align*}
-2\eps + \alpha \eps \ln(\alpha\eps/(1-\alpha))/4
& \geq -2\eps + \alpha\eps \ln(\alpha\eps)/4 - \alpha\eps\ln(1-p_t)^d/4 \\
& \geq -4\eps + \eps p_td/8 \geq \eps p_td/16
\end{align*}
An analogous bound holds for the upper tail, and a union bound completes the proof.
\end{proof}

Hence, we would like to estimate $|S_t(d)|$ and divide by $1-(1-p_t)^d$ to get estimates for
$n(d)$ (where $d$ is large). Our estimate for $|S_t(d)|$ is $\tC(d)$, and this
scaling is precisely what is done in \est.

\begin{definition} \label{def:geocdf}
\begin{asparaitem}
    \item $\cdf_{p_t,s}$: The cdf of $\tg_{p_t,s}$, formally $\cdf_{p_t,s}(k) = \Pr_{X \sim \tg_{p_t,s}}[X \leq k] = [1 - (1-p_t)^{k+1}]/[1-(1-p_t)^s]$.
    \item $\ell(d) = \floor{\EX_{X \sim \tg_{p_t,d}}[X]}$.
    \item $\rd(d) = d - \ell(d)$.
\end{asparaitem}
\end{definition}

Indeed, we will show that the ``ccdh" of
$|S_t(d)|$ is somewhat approximated by that of $\tC(d)$.

\begin{theorem} \label{thm:tc} $\EX[\sum_{r \geq d} \tC(d)] = \sum_{r \geq \rd(d)} \cdf_{p_t,r}(r-\rd(d)) \EX[|S_t(r)|]$.
\end{theorem}

\begin{proof} Note that $\ell(d)$ is monotonically increasing in $d$.
Any $v \in S$ such that $\ct_t(v) \geq \rd(d)$ will be counted as part of $\tC(r)$, for some
$r \geq d$. The quantity $\loss(v) = d_v - \ct_t(v)$, conditioned in $v \in S$, is distributed as $\tg_{p_t,d_v}$.
The probability of the loss being most $d_v - \rd(d)$ is exactly $\cdf_{p_t,d_v}(d_v - \rd(d))$.
\begin{eqnarray*}
    \EX[\sum_{r \geq d} \tC(d)] & = & \sum_v \Pr[v \in S] \Pr[\loss(v) \leq d_v - \rd(d) | v \in S] \\
    & = & \sum_v [1-(1-p_t)^{d_v}] \cdf_{p_t,d_v}(d_v - \rd(d)) \\
    & = &\sum_{r \geq \rd(d)} \cdf_{p_t,r}(r-\rd(d)) [1-(1-p_t)^d] n(r)
\end{eqnarray*}
By \Lem{highdeg}, $[1-(1-p_t)^r]n(r) = \EX[|S_t(r)|]$
\end{proof}

\subsection{Making sense of \Thm{tc}} \label{sec:tc}

Fix $p_t$ and $d$.
Consider $\cdf_{p_t,r}(r-\rd(d))$ as a function of $r$, and suppose it had value $0$ for $r < d$, and value $1$ for $r \geq d$.
Think of this as the ideal value for this function.
Then, by \Thm{tc}, $\EX[\sum_{r \geq d} \tC(d)] = \sum_{r \geq d}  \EX[|S_t(r)|]$, which
would be exactly what we want. We prove that the ``coefficients" $\cdf_{p_t,r}(r-\rd(d))$
behave like a step function with a transition roughly at $d$. So $\EX[\sum_{r \geq d} \tC(d)]$
is a sort of smoothed version of $\sum_{r \geq d} \EX[|S_t(r)|]$.

We begin with some approximations for $\cdf_{p_t,r}$. It is useful
to think of the limit as $p \rightarrow 0$
and reparametrize as $d = \mult/p_t$.
By \Lem{geo},
\begin{eqnarray*}
    \EX_{X \sim \tg_{p_t,d}}[X] & = & \frac{1 - p_t - (1-p_t)^{1+\mult/p_t} - \mult(1-p_t)^{\mult/p_t}}{p_t(1-(1-p_t)^{\mult/p_t})}\\
    & \approx & \frac{1 - p_t - e^{-\mult} - \mult e^{-\mult}}{p_t(1-e^{-\mult})} \\
    & \approx & \frac{1}{p_t} - \frac{\mult e^{-\mult}}{p_t(1-e^{-\mult}} \approx \frac{1}{p_t}(1-\mult e^{-\mult})
\end{eqnarray*}

Thus, $\rd(d) = \mult/p_t - 1/p_t + \mult e^{-\mult}/p_t$. Now consider some $r = x/p_t \geq \rd(p)$.
\begin{eqnarray}
    \cdf_{p_t,r}(r-\rd(d)) & = & \frac{1-(1-p_t)^{1+(1/p)\cdot(x- \mult + 1- \mult e^{-\mult})}}{1-(1-p_t)^{x/p}} \nonumber\\
    & \approx & \frac{1-\exp(-(x-\mult+1-\mult e^{-\mult}))}{1-\exp(-x)} \label{eqn:cdf}
\end{eqnarray}

Clearly, as $x$ becomes large, this expression goes to $1$. The minimum possible
value of $x$ is $\mult-1+\mult e^{\mult} $ (equivalently, $r = \rd(d)$), for which
the expression is $0$. It behaves roughly like a step function, with a transition point (roughly)
at $\mult-1+\mult e^{-\mult}$. As $\mult$ becomes large, the transition point is $\mult-1$, close
to $\mult$.
When $\mult$ is small, the extra $\mult e^{-\mult}$ additive terms
ensures the transition is closer to $\mult$. Of course, as $\mult$ becomes smaller, the function looks less
like a sharp transition function.
This is shown in \Fig{step}. We plot $\cdf_{p_t,r}$ according to \Eqn{cdf} for $k=5, 10, 100$. The red vertical
line is $x=k$ (so $r = k/p_t$), and we draw dashed vertical lines corresponding to value $0.1$ and $0.9$.
The width between the dashed lines is a rough measure of the error in approximation.
Observe how it is fairly close to a step function for $k=10$, and is a coarser approximation for $k=5$.

\begin{figure*}[t]
    \centering
    \begin{subfigure}[b]{0.3\textwidth}
        \centering \includegraphics[width=\textwidth]{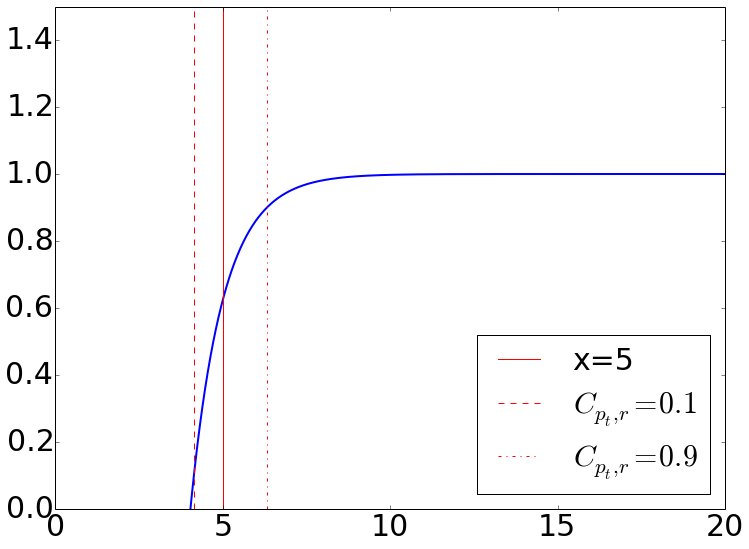}
        \caption{$k=5$.}
        \label{fig:logisticC-k5}
    \end{subfigure}
    \hfill
    \begin{subfigure}[b]{0.3\textwidth}
        \centering \includegraphics[width=\textwidth]{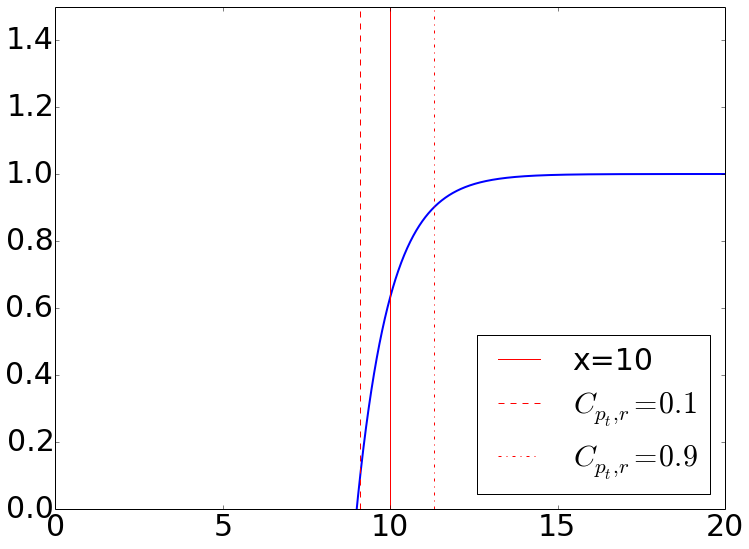}
        \caption{$k=10$.}
        \label{fig:logisticC-k10}
    \end{subfigure}
    \hfill
    \begin{subfigure}[b]{0.3\textwidth}
        \centering \includegraphics[width=\textwidth]{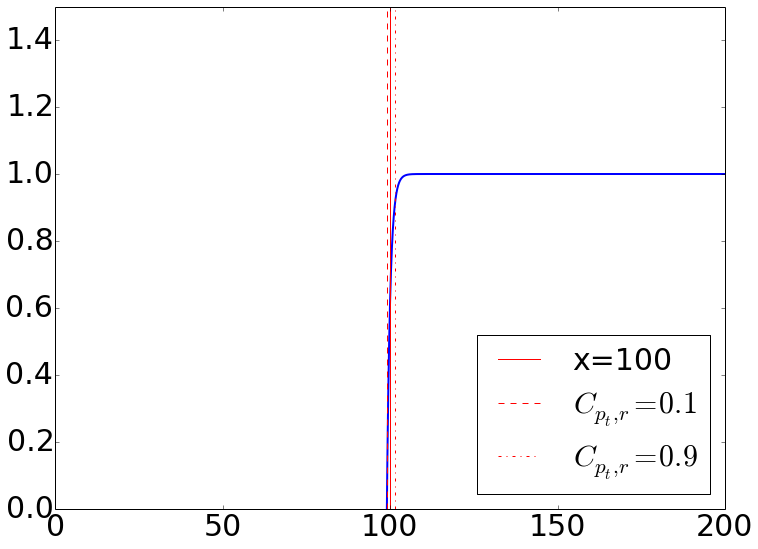}
        \caption{$k=100$.}
        \label{fig:logisticC-k100}
    \end{subfigure}
    \caption{Plots of $\cdf_{p_t,r}$ according to~\Eqn{cdf} for different values of $k$. Note that $r$ is set to $x/k$. In each
    plot, the thin vertical line is $x = k$, and the dashed and dotted lines
correspond to values of $0.01$ and $0.09$, respectively.}
    \label{fig:step}
\end{figure*}

%
%
% In
% Figures~\ref{fig:logisticC-k5},\ref{fig:logisticC-k1},\ref{fig:logisticC-k10},\ref{fig:logisticC-k100},
% we plot the function for different fixed values of $\mult$.  We also plot the
% vertical lines $x=\mult$, and the points the expression equals $0.1$ and $0.9$.

% \begin{figure}[!h]%[!t]
% \centering \includegraphics[width=3.5in]{logisticC-k5}
% \caption{Logistic behavior of $C_{p_t, r}$ as described in (\ref{eq:Cr}).  We
% fix $k=5$ and plot the function over $x$.  Additionally we draw attention to
% approximate transition point by plotting vertical lines at $x=k$, and the points
% $x$ at which $C_{p_t, r} = 0.1$ and $0.9$.}
% \label{fig:logisticC-k5}
% \end{figure}
%
% \begin{figure}[!h]%[!t]
% \centering \includegraphics[width=3.5in]{logisticC-k1}
% \caption{Similar to Figure~\ref{fig:logisticC-k5} with $k=1$.}
% \label{fig:logisticC-k1}
% \end{figure}
%
% \begin{figure}[!h]%[!t]
% \centering \includegraphics[width=3.5in]{logisticC-k10}
% \caption{Similar to Figure~\ref{fig:logisticC-k5} with $k=10$.}
% \label{fig:logisticC-k10}
% \end{figure}
%
% \begin{figure}[!h]%[!t]
% \centering \includegraphics[width=3.5in]{logisticC-k100}
% \caption{Similar to Figure~\ref{fig:logisticC-k5} with $k=100$.}
% \label{fig:logisticC-k100}
% \end{figure}

Hence, $\EX[\sum_{r \geq d}\tC(r)]$ is much further from $\EX[\sum_{r \geq d}|S_t(r))|]$,
and \est{} provides worse results. But we set $d_{thr} > \log(1/\eps)/p_t$. So for degrees
close to $1/p_t$, we do not use the tail estimator.

\section{The Relative Hausdorff distance} \label{sec:rh}

One of the main challenges in experimentally validating the behavior of \est{} is
in defining a distance between ccdhs. As we hinted earlier, existing statistical distances
do not capture ``similarity" of ccdhs.
Motivated by concerns (detailed below), we define a new notion of distance between ccdhs (technically,
between cumulative complementary histograms).
This is inspired by the geometric notion of Hausdorff distance between subsets of a metric space.
We say a ccdh is non-trivial
if it contains some non-zero point.

\begin{definition} \label{def:hauss} Let $F$ and $G$ be non-trivial ccdhs. Fix non-negative numbers $\eps, \delta$.
The distributions $F$ and $G$ are \emph{$(\eps,\delta)$-close by Relative Hausdorff (RH) distance} if:
$$ \forall d, \exists d' \in [(1-\eps)d, (1+\eps)d], \ \textrm{such that} \ |F(d) - G(d')| \leq \delta F(d).$$
(An analogous condition holds with $F$ and $G$ switched.)

The \emph{RH}-distance between $F$ and $G$ (denoted $RH(F,G)$) is  $\inf\{\eps | \textrm{$F$ and $G$ are $(\eps,\eps)$-close}\}$.
\end{definition}

Note that the $RH$-distance can be greater than $1$.
For $\eps' \geq \eps$ and $\delta' > \delta$, if $F$ and $G$ are $(\eps,\delta)$-close,
they are also $(\eps',\delta')$-close. Since $F$ and $G$ are non-trivial, we can set $\eps$ to be large enough so
that for some $\delta$, $F$ and $G$ are $(\eps,\delta)$-close. Thus, the RH distance always exists.
If $RH(F,G) = 0$, then $F$ and $G$ are identical.

Observe that RH distance tolerates error both in degree and frequency, which is very important for
comparing degree distributions.
The RH distance exactly captures the notion of being close in log-scale, but is a much more stringent condition.
It forces \emph{all} points in $F$ to be close to some point in $G$ (and vice versa).
All ``outlier" and tail behavior in $F$ must be approximated in $G$. For RH-close ccdhs, the maximum degrees
must be close, and furthermore, there must be approximate agreement for frequencies at all scales.

To understand numerics, we think it is useful think of an RH-distance $< 0.05$ to be quite small.
Suppose $RH(N,\hN) < 0.05$ for a true ccdh $N$ and our algorithm output $\hN$. This means that
for every reported point $\hN(d)$ is within 5\% of some $N(d')$, where $d'$ is within 5\% of $d$
(and vice versa). Any RH distance greater than $1$ is very large, since we only get closeness
when $\eps \geq 1$.

\subsection{Problems with KS-statistic}

Fix two ccdhs $F$ and $G$. A standard comparison metric is the Kolmogorov-Smirnov (KS) statistic, $KS(F,G) = \max_x |\tF(x) - \tG(x)|$,
where $\tF, \tG$ are normalized as distributions. (So $\tF(d)$ is the fraction of vertices
with degree at least $d$.)

We discuss specific problems with the KS statistic and show how RH avoids these pitfalls.
(The exact same issues also holds for normed distances, so we do not explicitly calculate these.)

\textbf{Comparing cliques:} Let $F$ be the ccdh of an $n$-clique and $G$ be the ccdh of an $(n-1)$-clique.
So $\forall \ 0 \leq i \leq n-1, F(i) = n$, $\forall \ 0 \leq i \leq n-2$, $G(i) = n-1$, and all other values are $0$.
The KS-statistic is actually $1$ (which is extremely large), since $\tG(n-2) = 0$ but $\tF(n-2) = 1$.
This is inconsistent with our intuitive notion that these degree distributions are similar.
The RH distance is $O(1/n)$, since it allows for error in degree and frequency.

\textbf{Star vs matching:} Let $F$ be the ccdh of a star with $n$ vertices, and $G$
be the ccdh of a matching (disjoint edges) with $n$ vertices. (Assume $n$ is even.)
So $F(1) = n$, $\forall 2 \leq i \leq n-1$, $F(i) = 1$, and other values are $0$.
We also have $G(1) = n$, and all other values are $0$. The values of $F$ that are $1$
are insignificant compared to the dominant $F(1) = n$. A calculation shows $KS(F,G) = O(1/n)$,
though we should probably consider them different.
On the other hand, $RH(F,G) = 1 - \Theta(1/n)$. The ``outlier" $F(n-1)=1$ forces the $\eps$ to be $1-\Theta(1/n)$,
since $G(i) = 0$ for $i > 1$.

\textbf{Ignoring the tail:} Let $F$ be the ccdh of the as-Skitter graph, as plotted in \Fig{skitter-ccdh}.
Let $G$ be the same ccdh up to degree $100$ and zero afterwards. In other words, $G$ is identical to $F$
up to the ``tail" starting at degree $100$. The fraction of vertices with degree $> 100$ is at most $0.01$.
A calculation shows that $KS(F,G) < 0.01$. So ignoring a large portion of the tail still yields
small KS-distance. The RH-distance is $>0.99$, since $\eps$ needs to be large to handle
the tail of $F$.

\section{Experimental Results} \label{sec:expresults}
We implemented the algorithm in Python and performed experiments on a Samsung
NP-QX411L laptop with an Intel Core i5-2450M 2.5GHz four core processor and
5.7GB of memory.  To simulate a stream, we convert a graph to a list of edges
stored in a text file, and read the file one line at a time.  In the case that
the graph is directed, we treat it as undirected by considering each edge as an
unordered pair of vertices.  Note that this may imply multi- or parallel edges,
though we calculate degrees for the actual ccdh respecting this notion.

We test the algorithm on a number of graphs from the SNAP~\cite{snap} and
KONECT~\cite{konect} collections, the statistics of which are summarized in
Table~\ref{table:manygraphs}.  We use the as-Skitter graph on 1.7M nodes and 11M
edges as a case study.

We use the phrase \emph{storage of \degdist} to indicate the total storage $|S_h| + |S_t|$.
As explained in \Thm{time}, this depends on $p_h$ and $p_t$.

\subsection{Convergence of \degdist} \label{sec:conv}
We demonstrate how increasing the storage of \degdist{} leads to convergence of
the ccdh. We fix the as-Skitter graph. We increase the storage by letting $p_h$
range from $0.01$ to $0.1$ in increments of $0.01$, and $p_t$ range from $0.01$
to $0.16$ in increments of $0.01$.  For each setting of $p_h$ and $p_t$, we
perform five independent runs of \degdist. We also run ten independent runs
fixing $p_h=0.005, p_t=0.01$.  For each such run, we compute the RH distance
between the output of \degdist{} with the true ccdh. The results are shown
in~\Fig{skitter_scatter}. Observe how the RH distance goes to zero as the
storage increases.  In particular, \degdist{} outputs a ccdh with RH distance as
small as $0.03$ using 230K space.
%
% An important general observation is that the RH distance improves with the
% size of the sample set.  Figure~\ref{fig:skitter_scatter} plots the RH
% distance of estimates as the combined sizes of the sample sets $S_h, S_t$
% increases.  In each experiment, we fix $p_h=0.01$ and vary $p_t$.  We measure
% the RH distance of the estimate by determining the minimum $\eps$ such that
% the estimate is $(\eps, \eps)$-close the the actual ccdh, beginning with
% $\eps=0.001$.  Each marker corresponds to one experiment, and is plotted
% according to space used for the estimator against RH distance.  We see that as
% the space increases, the RH distances tend to converge.
\begin{figure}
\centering \includegraphics[width=3.0in]{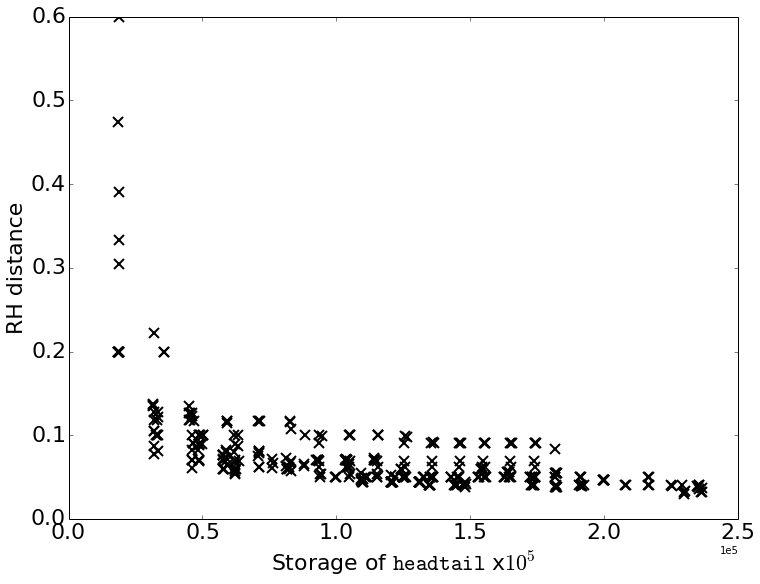}
\caption{RH distance as the storage of \degdist{} increase.}
\label{fig:skitter_scatter}
\end{figure}

We do a more nuanced study of how $p_h$ and $p_t$ affect convergence.
In this experiment, we fix a value $p_h$ and vary $p_t$ in
increments of $0.02$.  We repeat this process for $p_h = 0.01, 0.025, 0.05,
0.075, 0.1$.  The RH distances of the runs are plotted in
Figure~\ref{fig:skitter_summary_varytail}.  Each line in the plot corresponds to
a fixed $p_h$ value, and the RH distances are plotted against $p_t$.  We point
out that an RH distance of about $0.04$ is achieved with head and tail
probabilities as small as $0.025, 0.03$, respectively, resulting in a total
sample size of 82K or $0.7\%$ of the edge stream.
For each fixed $p_h$, increasing $p_t$ initially decreases the RH distance, but it eventually
converges to a non-zero value. This is because all the error is coming from the head estimate.
As we increase $p_h$, the convergence value goes down to zero, as expected.
%
%  We also note that varying the
% tail probability has a more fine-tuned effect on the RH distance than varying
% the head probability.  That is, in general (with the exception of $p_h = 0.01$),
% each of the lines in the plot tend around the same RH distance as $p_t$
% increases.
\begin{figure}
\centering \includegraphics[width=3.0in]{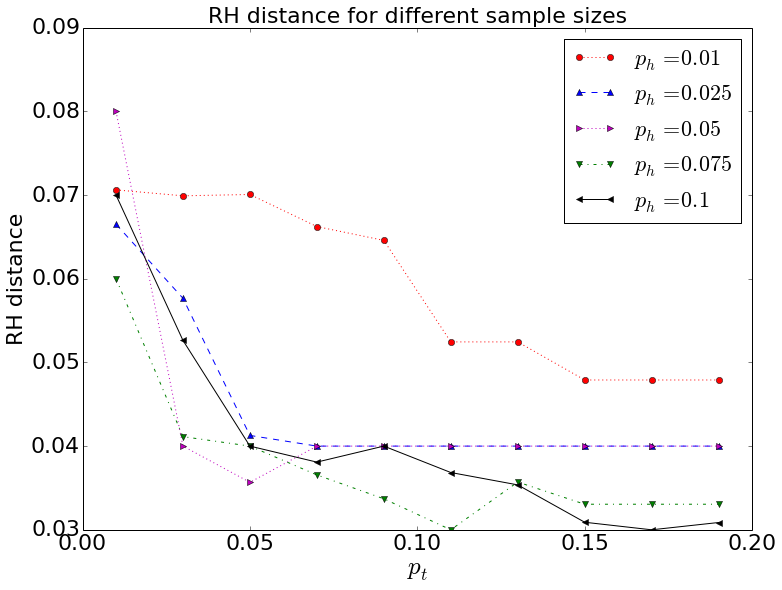}
\captionsetup{width=0.9\textwidth}
\caption{RH distance of the estimate output by our algorithm as $p_h$ and $p_t$
vary.  Each line in the plot correponds to a fixed value for $p_h$, and plots
the RH distance as $p_t$ varies.  A near optimal RH value is achived with $p_h =
0.025$ and $p_t = 0.03$, which yielded sample sets with $|S_h|+|S_t| \approx
0.007m$.}
\label{fig:skitter_summary_varytail}
\end{figure}

\subsection{Results for various graphs}
Here we demonstrate the quality of the estimates output by \degdist~on a variety
of graphs.  Each of the graphs are from the SNAP graph collection~\cite{snap}
with the exception of the youtube and youtube-friendship graphs which are from
the KONECT~\cite{konect} collection.  The node and edge set sizes of each graph
are given in the second and third columns of Table~\ref{table:manygraphs},
respectively.  For each graph we include the storage of the
algorithm and the RH distance of the estimate for two example runs.
The storage is less than $< 1\%$ in almost at runs, and certainly less than $< 2\%$.
Observe how the RH distance is usually less than $0.1$. In our worst examples,
(soc-Pokec and com-Orkut), the RH distance is less than $0.15$. We stress that RH distance
is a rather stringent condition, since it requires closeness of the estimate
at \emph{all} degrees.

In \Fig{ccdh} of the introduction, we have plotted the actually ccdh and the output
of \degdist{} for three of these graphs.
Observe the near identical match in all examples.

\begin{table}
\small
\centering
\begin{tabular}{l|c|c|r|l}
\hline
Graph & $n$ & $m$ & Space & RH distance\\
\hline\hline
youtube & 1.1M & 3M & 21K & $0.1$\\
%& & & 56K & $0.093$\\
& & & 90K & $0.076$\\\hline
%& & & 113K & $0.071$\\\hline
wiki-Talk & 2.3M & 5M & 38K & $0.1$\\
& & & 74K & $0.055$\\\hline
%& & & 122K & $0.05$\\
%& & & 170K & $0.05$\\\hline
youtube-friendship & 3M & 9M & 80K & $0.067$\\
& & & 196K & $0.05$\\\hline
as-Skitter & 1.7M & 11M & 31K & $0.1$\\
& & & 69K & $0.073$\\\hline
%& & & 123K & $0.05$\\\hline
%& & & 169K & $0.05$\\\hline
soc-Pokec & 1.6M & 30M% & 31K & $0.5$\\
& 75K & $0.29$\\
% & & & 147K & $0.34$\\
& & & 212K & $0.14$\\\hline
com-LiveJournal & 4M & 34M% & 80K & $0.2$\\
%& & & 184K & $0.1$\\
& 335K & $0.08$\\
& & & 467K & $0.058$\\\hline
com-Orkut & 3M & 117M% & 61K & $0.67$\\
%& & & 144K & $0.19$\\
& 273K & $0.14$\\
& & & 387K & $0.13$\\\hline
\end{tabular}
\caption{Performance of \degdist~for a number of graphs.}
\label{table:manygraphs}
\end{table}

\subsection{Errors at different scales}\label{sec:differentscales}
Here we investigate how well \degdist{} performs at different scales.
Specifically, we measure the error of a ccdh estimate at each degree.  Let $N$
be the ccdh of the as-Skitter graph, and $\hN$ be the \degdist{} output.
The RH distance is maximized over all degrees, so we do a more detailed
analysis of the estimate errors.
We fix a value for $\eps$ and for each
degree $d$ compute the minimum value $\delta$ such that $\exists d' \in
[(1-\eps)d, (1+\eps)d]$ where $|N(d) - \hN(d')| \leq \delta N(d)$ and vice versa.
In words, we are ``opening up'' the definition of RH-distance and looking at the
profile for every degree.

We performed a run of \degdist{} with $p_h = 0.01$ and $p_t=0.0007$ for the as-Skitter
graph. This used a storage of 31K ($< 0.5\%$ of stream). We then plot in \Fig{allscales}
the corresponding $\delta$ values with $\eps$ set to $0.1$ .
The red `\textsf{x}' markers denote the $\delta$-values for \degdist{} (the
other markers are explained later). Observe how the $\delta$ values
are quite small throughout, and peak at degree $100$ to roughly $0.08$.
In this case, \degdist{} achieves an RH-distance of about $0.1$ with 31K space.

\begin{figure}
\centering \includegraphics[width=3.0in]{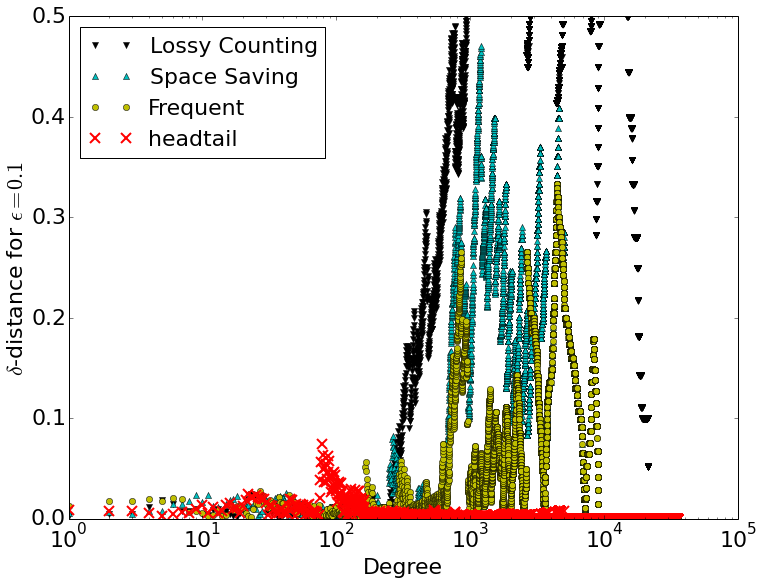}
\captionsetup{width=0.9\textwidth}
\caption{RH distances at different degrees.  We plot the $\delta$-distance for
$\eps=0.1$.  The red `\textsf{x}' markers correspond to an estimate output
by~\degdist{} using a storage of 31K.  The estimate is $(0.1, 0.08)$-far from
the true ccdh.  The rest of the plots correspond to combinations of the head
estimator using 17K space and the heavy hitter algorithms using 34K space for a
total of 51K space.  The \emph{lossy counting} estimate is $(0.1, 1.5)$-far from
the true ccdh, the \emph{space saving} estimate $(0.1, 0.4)$-far and the
\emph{frequent} estimate is $(0.1, 0.33)$-far from the true ccdh.}
\label{fig:allscales}
\end{figure}

%\begin{figure*}[ht]
%\centering
%  \begin{minipage}[b]{0.5\linewidth}
%    \centering
%    \includegraphics[width=\linewidth]{skitter_varytail_eps05_space0186}
%    \label{vt:skitter05_0186}
%    \caption{An estimate using 31K space.  Let $\eps=0.05$, the estimate is
%$(0.05, 0.13)$-far from the true ccdh.}
%    %\vspace{4ex}
%  \end{minipage}%%
%  \begin{minipage}[b]{0.5\linewidth}
%    \centering
%    \includegraphics[width=\linewidth]{skitter_varytail_eps05_space1016}
%    \label{vt:skitter05_1016}
%    \caption{An estimate using 172K space.  Let $\eps=0.05$, the estimate is
%$(0.05, 0.08)$-far from the true ccdh.}
%    %\vspace{4ex}
%  \end{minipage}
%  \begin{minipage}[b]{0.5\linewidth}
%    \centering
%    \includegraphics[width=\linewidth]{skitter_varytail_eps1_space0186}
%    \label{vt:skitter1_0186}
%    \caption{An estimate using 31K space.  Let $\eps=0.1$, the estimate is
%$(0.1, 0.08)$-far from the true ccdh.}
%    %\vspace{4ex}
%  \end{minipage}%%
%  \begin{minipage}[b]{0.5\linewidth}
%    \centering
%    \includegraphics[width=\linewidth]{skitter_varytail_eps1_space1016}
%    \label{vt:skitter1_1016}
%    \caption{An estimate using 172K space.  Let $\eps=0.1$, the estimate is
%$(0.1, 0.03)$-far from the true ccdh.}
%    %\vspace{4ex}
%  \end{minipage}
%\caption{The error at different degrees for estimates using 31K space and 172K
%space.  In Figures~\ref{vt:skitter05_1086} and~\ref{vt:skitter05_1016}, we fix
%$\eps = 0.05$ and in Figures~\ref{vt:skitter1_0186} and~\ref{vt:skitter1_1016}
%we fix $\eps-0.1$.}
%\label{fig:varytail}
%\end{figure*}

\subsection{Comparing to other methods}\label{sec:compare}

While there is no existing small-space algorithm that has demonstrable convergence
to the ccdh, there are numerous algorithms to only capture the tail.
These are classic ``heavy hitters" algorithms: the \emph{frequent} algorithm~\cite{demaine2002frequency,
  karp2003simple,BerindeICS10}, the \emph{lossy counting}
algorithm~\cite{manku2002approximate}, and the \emph{space saving}
algorithm~\cite{metwally2005efficient}. We study the performance of these methods.
For convenience, we use ``head estimator" to denote the algorithm that simply
takes uniform samples of vertices and uses their degrees to estimate the full ccdh.
This is basically what \degdist{} employs for $d \leq d_{thr}$.

We fix the as-Skitter graph, and set the storage used by these algorithms to 35K. (Note that with storage 31K, \degdist{} gives an estimate
with RH-distance less than $0.1$.) We show the resulting estimates of these algorithms
in \Fig{compare}. Not surprisingly, none of these algorithms give reasonable estimates
for $N(d)$, where $d \leq 10^3$.

At the face of it, the above algorithms perform reasonably well on the tail. The head estimator
(which is quite simple) seems to work well for the head. Could we just combine these algorithms,
and outperform \degdist? We show that this is not the case. Crucially,
none of these algorithms actually get accurate estimates even
at the moderate to high degrees, despite the apparent closeness in the log-log plot of \Fig{compare}.

We convert the existing algorithms for the full ccdh, by combining with the head estimator.
Pick (say) the algorithm \emph{frequent}.
We first run the head estimator with 20K space. We choose an appropriate $d_{thr}$, where we apply the head
estimator for $d \leq d_{thr}$, and \emph{frequent} for $d > d_{thr}$.
We pick the $d_{thr}$ that minimizes the RH distance to $\{N(d)\}$.
We do the same for each of \emph{frequent}, \emph{space saving}, and \emph{lossy counting}.
Note that we are being extra generous to the competing methods. First, the total
storage used is about 50K. Furthermore, we choose the $d_{thr}$ to minimize RH distance,
while \degdist{} chooses it based on a fixed formula.

The RH distance we achieved was $0.3$ (\emph{frequent}), $0.5$ (\emph{space saving}), and $1.5$ (\emph{lossy counting}).
All of these used storage 50K. In contrast, \degdist{} had RH distance of $0.1$ with 31K storage.
We measure the errors at all scales in~\Fig{allscales}, for all these algorithms. This is exactly
using the explanation in previous section, by setting $\eps = 0.1$, and plotting the $\delta$
values for all the estimates.

We immediately see how the $\delta$-values (errors) for all the competing procedures are much higher than \degdist.
Indeed, for degrees around $10^3$, the errors of the other procedures are extremely high, despite higher storage.
We see that \degdist{} handily beats all the procedures, at pretty much all scales simultaneously.
In \Fig{frequent_head}, we plot the output ccdh for the head estimator combined with \emph{frequent}.
As expected from \Fig{allscales}, we see a fair amount of fluctuation from the true ccdh in the the intermediate
to high degrees. We stress that a small fluctuation in a log-log plot is actually a fairly large error
in the RH measure.

For completeness, we increase the storage of the competing methods to get RH distance of around $0.1$.
For all the other algorithms, we require storage more than 150K to get comparable
error to what \degdist{} gives with 31K storage.

%  We started by using sample
% space comparable to what we use for~\degdist, and then increase the space to see
% how well the estimators perform.  When combining with the head estimator, we
% chose $d_{thr}$ to produce a ccdh estimate as good as possible with respect to
% RH distance.  In every case it resulted in a threshold much higher than that
% used in our algorithm, so that most of the ccdh is captured by the head
% estimator.
%
%
% In this section we compare the results of our algorithm to some proposed and
% existing alternatives for computing the ccdh.  First, we use the head estimator
% alone.  Then we implement three ``heavy hitters'' algorithms
% -  (see Section~\ref{sec:relatedwork}).  We
% plot the estimated ccdhs of each of these methods in Figure~\ref{fig:compare}.
%
\begin{figure*}
\centering
  \begin{minipage}[b]{0.5\linewidth}
    \centering
    \includegraphics[width=0.9\linewidth]{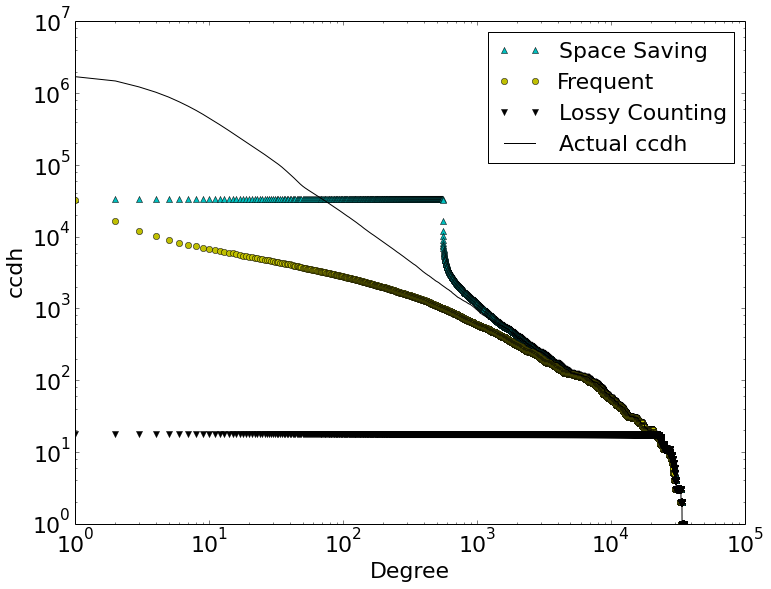}
    \captionsetup{width=0.9\textwidth}
    \caption{ccdh estimates output by the \emph{frequent}, \emph{lossy counting},
and \emph{space saving} algorithms each using a storage of 35K.\\}
    \label{fig:compare}
    %\vspace{4ex}
  \end{minipage}%%
  \begin{minipage}[b]{0.5\linewidth}
    \centering
    \includegraphics[width=0.9\linewidth]{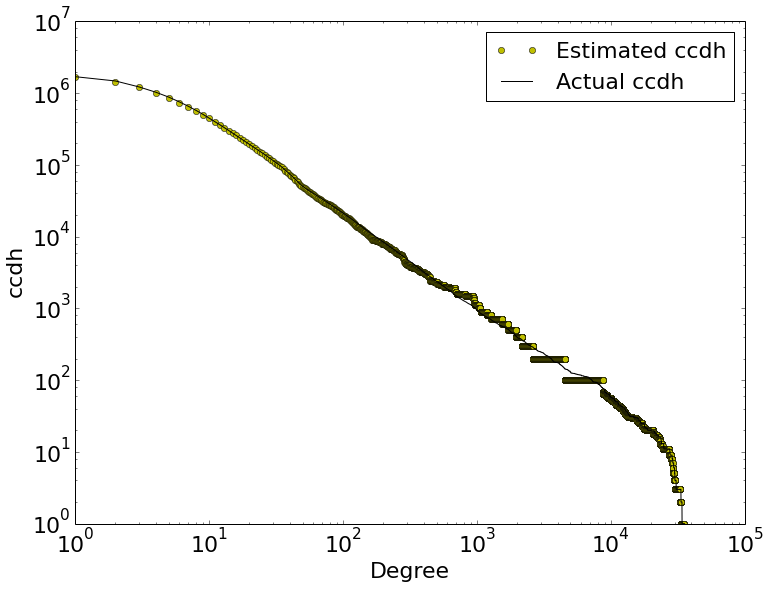}
    \captionsetup{width=0.9\textwidth}
    \caption{ccdh estimate output by the head estimator combined with the
\emph{frequent} algorithm using a storage of 50K.  The RH distance
is $0.33$.}
    \label{fig:frequent_head}
    %\vspace{4ex}
  \end{minipage}
\end{figure*}

\subsection{Results for different stream orderings}
As stated previously, our algorithms do not assume any stream order.  In this
section we test the performance of the algorithm when provided the stream in
different orderings.  We use six different orderings in total.  The first three
are different random orderings.  The second three are each edgelists (that is,
all the edges adjacent to a particular node are read in sequence), but the
orderings of the nodes are different.  In one, we read the nodes of highest
degree first, in another we read the nodes in increasing order of degree, and in
the last we consider a random ordering of the nodes.  In each experiment we let
$p_h = 0.01$ and $p_t = 0.04$.  The standard deviation of the RH distances for
each ordering is $0.009$.  Table~\ref{table:ordering} summarizes the RH distance
of estimated ccdhs with different stream orderings.

\begin{table}
\small
\centering
\begin{tabular}{l|l}
\hline
Ordering & RH distance\\
\hline\hline
Random1 & $0.068$\\\hline
Random2 & $0.06$\\\hline
Random3 & $0.07$\\\hline
Edgelist: Decreasing order of degree & $0.08$\\\hline
Edgelist: Increasing order of degree & $0.083$\\\hline
Edgelist: Random & $0.061$\\\hline
\end{tabular}\captionsetup{width=0.9\textwidth}
\caption{Performance of \degdist~for different stream orderings.  The first
three are different random stream orderings.  The second three are edgelists
permuted by the nodes.  In each trial $p_h = 0.01, p_t = 0.04$.}
\label{table:ordering}
\end{table}

\section*{Acknowledgment}

The authors would like to thank Tammy Kolda, Ali Pinar, and David Mayer for useful discussions.
Much of this work was done in Sandia National Laboratories, Livermore, and funded by the DARPA
GRAPHS program.



%\bibliographystyle{IEEEtran}


% \end{thebibliography}

\bibliographystyle{IEEEtran}
\bibliography{streaming-degree,streaming_triangles}

\end{document}